\documentclass[reqno,11pt]{amsart}
\usepackage{amsmath,amsfonts,amsgen,amstext,amsbsy,amsopn,amsthm, amssymb}

\newtheorem{Lemma}{Lemma} \newtheorem{Theorem}{THEOREM}
\newtheorem{prop}{Proposition} \newtheorem{assu}{Assumption}
 \theoremstyle{definition}

 \newcommand\calC{{\mathcal{C}}}
\newcommand\Tr{{\rm Tr\,}}
\newcommand\Trs{{\rm Tr}_0\,}
\newcommand\id{\mathbb{I}}
\newcommand\nn\nonumber

\numberwithin{equation}{section}

 \newcommand{\R}{\mathbb{R}}
 \newcommand{\N}{\mathbb{N}}
\newcommand{\Z}{\mathbb{Z}} \newcommand{\C}{\mathbb{C}}
 \newcommand{\F}{\mathcal{E}}
\newcommand{\E}{\mathcal{E}}

 \newcommand\infspec{{\rm{inf\, spec\,}}}
 \newcommand\eb{E_{\rm b}}
\newcommand\as{a_{\rm s}}

 \DeclareMathOperator{\dive}{div}
 \DeclareMathOperator{\const}{const}
 \DeclareMathOperator{\re}{Re}

\begin{document}

\title[Low Density Limit of BCS Theory]{Low Density Limit of BCS
  Theory and \\ Bose-Einstein Condensation of Fermion Pairs}

\author[C. Hainzl]{Christian Hainzl} \address{{\rm (Christian Hainzl)}
  Mathematisches Institut, Universit\"at T\"ubingen, Auf der
  Morgenstelle 10, 72076 T\"ubingen, Germany}
\email{christian.hainzl@uni-tuebingen.de}

\author[R. Seiringer]{Robert Seiringer} \address{{\rm (Robert
    Seiringer)} Department of Mathematics and Statistics, McGill
  University, 805 Sherbrooke Street West, Montreal, QC H3A 2K6,
  Canada} \email{rseiring@math.mcgill.ca}

\date{May 5, 2011}

\begin{abstract}
  We consider the low density limit of a Fermi gas in the BCS
  approximation. We show that if the interaction potential allows for
  a two-particle bound state, the system at zero temperature is well approximated by the
  Gross-Pitaevskii functional, describing a Bose-Einstein condensate
  of fermion pairs.
\end{abstract}

\maketitle

\section{Introduction and Main Results}

\subsection{Introduction}

The bosonic behavior of pairs of fermions is a topic that has been
investigated in condensed matter physics for more than half a
century. It plays a crucial role in the Bardeen-Cooper-Schrieffer
(BCS) theory of superconductivity \cite{BCS}, and is used to explain
the superfluid behavior of He$^3$ and cold gases of fermionic atoms,
for instance. As long as the pair of fermions is tightly bound, it is
not surprising that it effectively behaves like a boson, and hence can
form a Bose-Einstein condensate (BEC). In BCS theory the pairing
mechanism is also important under very weak attraction where the
separation of the paired particles can be much larger than the average
particle spacing, however.

It was realized in the eighties \cite{Leg,nozieres} that BCS theory
actually applies both to the case of BECs of tightly bound fermions
and to cases where the pairing mechanism is very weak. The regime
in-between is called the BEC-BCS crossover regime
\cite{randeria,zwerger}. This crossover is currently a hot topic in
condensed matter physics, and recent experiments on cold atomic gases
have been able to probe large parts of this regime. We refer to
\cite{BDZ} for a recent review.

From the mathematical physics point of view, the pairing mechanism in
fermionic systems is poorly understood, and there are no rigorous
results starting from first principle, i.e., with an appropriate
many-body Hamiltonian. In this work, we shall {\em assume} the BCS
approximation to be correct, and investigate some of its
consequences. This paper can be viewed as a sequel to the recent work
\cite{FHSS,FHSS2} where the emergence of Ginzburg-Landau (GL) theory
\cite{GL} from BCS theory was studied. Close to the critical
temperature, GL arises as an effective theory on the macroscopic
scale, describing the variations in the density of fermion pairs. For
this it is not necessary to form actual bound states between the
fermions, a very weak attraction is sufficient for pairing.

In this paper we are interested in the low density limit in the case where the 
interparticle interaction {\em does} allow for two-particle bound states. We 
consider the system at zero temperature and show that the macroscopic variations 
in the pair density are, to leading order, correctly described by the 
Gross-Pitaevskii (GP) functional \cite{gross,pitaevskii}, describing a BEC of 
fermion pairs with an effective {\em repulsive} interaction. We allow for a large 
class of possible interactions among the particles. Essentially our sole 
assumption will be the existence of two-particle bound states.  The existence of 
such bound states is crucial here. In their absence one obtains an ideal Fermi gas 
in the low density limit, up to exponentially small corrections 
\cite{FHNS,HS,HS3}. 

The proof of our results uses the same tools as the proof of the main
theorem in \cite{FHSS}. Parts of it are simpler, in fact, since we
work at zero temperature here. We shall demonstrate that the
semiclassical estimates in \cite{FHSS} extend to the zero temperature
case.

\subsection{The BCS Functional}

We consider a macroscopic sample of a system of spin $1/2$
fermions at zero temperature. For simplicity, we restrict our
attention to three spatial dimensions, but our analysis applies to any
dimension $1\leq d\leq 3$.  The interaction among the fermions is
described by a local two-body potential $V$. In addition, the
particles are subject to external electric and/or magnetic
fields. Neutral atoms would not couple to these fields, of course, but
there can be other forces, e.g., arising from rotation, with a similar
mathematical description.  In BCS theory the state of the system is
described in terms of a $2\times 2$ operator valued matrix
\begin{equation}\label{def:gamma}
  \Gamma = \left( \begin{array}{cc} \gamma & \alpha \\ 
\bar\alpha & 1 -\bar\gamma \end{array}\right) 
\end{equation}
satisfying $0\leq \Gamma \leq 1$ as an operator on $L^2(\R^3)\oplus
L^2(\R^3)$. The bar denotes complex conjugation, i.e., $\bar \alpha$
has the integral kernel $\overline {\alpha(x,y)}$. The fact that
$\Gamma$ is hermitian implies that $\gamma$ is hermitian and $\alpha$
is symmetric, i.e., $\gamma(x,y) = \overline{\gamma(y,x)}$ and
$\alpha(x,y)=\alpha(y,x)$. Moreover, since $\Gamma^2\leq \Gamma$, we
have $0\leq \gamma\leq 1$ and $0\leq \alpha\bar \alpha \leq
\gamma(1-\gamma)$.

We are interested in the effect of weak and slowly varying external
fields. Hence we introduce a small parameter $h>0$ and write the
external magnetic and electric potentials as $h A(h x)$ and $h^2
W(hx)$, respectively. In order to avoid having to introduce boundary
conditions, we assume that the system is infinite and periodic with
period $h^{-1}$, in all three directions. In particular, $A$ and $W$
are periodic, and we assume that the state $\Gamma$ is periodic. Our
goal is to calculate the ground state energy per unit volume, and the
corresponding BCS minimizer.

We find it convenient to do a rescaling and use macroscopic variables
instead of the microscopic ones. The rescaled BCS functional has the
form
\begin{align}\nonumber
  \F^{\rm BCS}(\Gamma) & := \Tr \left[ \left( \left(-i h \nabla + h
        A(x) \right)^2 -\mu + h^2 W(x)\right) \gamma \right] \\ &
  \quad \ + \int_{\calC\times \R^3} V(h^{-1}(x-y)) |\alpha(x,y)|^2 \,
  { dx \, dy} \label{def:bcs}
\end{align}
where $\calC$ denotes the unit cube $[0,1]^3$, and $\Tr$ stands for
the trace per unit volume. Explicitly, if $\chi$ denotes the characteristic
function of $\calC$, and $B$ is a periodic operator with $\chi B \chi$
trace class, $\Tr B$ equals the usual trace of $\chi B\chi$. The
location of the cube is obviously of no importance. Using the Floquet
decomposition \cite[Sect.~XIII.16]{ReSi}, it is not difficult to see that the
trace per unit volume has the usual properties of a trace like
cyclicity, for instance, and standard inequalities like H\"older's
inequality hold. This is discussed in detail in~\cite[Sect.~3]{FHSS}.

In (\ref{def:bcs}) we choose units such that the particle mass equals
$1$. The particles have spin 1/2, which adds an extra factor $2$ to
the energy. The chemical potential is denoted by $\mu/2$, for
convenience, and the external electric potential is really $W /2$.

For heuristic arguments explaining the derivation of the BCS
functional (\ref{def:bcs}), we refer to \cite[Appendix~A]{HHSS}.  The
BCS state of the system is a minimizer of this functional over all
admissible states $\Gamma$, i.e., periodic $\Gamma$ of the form
(\ref{def:gamma}) satisfying $0\leq \Gamma\leq 1$.

We make the following assumptions on the potentials $A$ and $W$ in
(\ref{def:bcs}). Our results presumably hold under slightly weaker
regularity assumptions on $W$ and $A$, but to keep things simple we
shall not aim for the weakest possible conditions.

\begin{assu}\label{as0}
  We assume both $W$ and $A$ to be periodic with period $1$. We
  further assume that $\widehat W(p)$ and $|\widehat A(p)|(1+|p|)$ are
  summable, with $\widehat W(p)$ and $\widehat A(p)$ denoting the
  Fourier coefficients of $W$ and $A$, respectively. In particular,
  $W\in C^0(\R^3)$ and $A\in C^1(\R^3)$.
\end{assu}

The interaction potential will be assumed to satisfy the following
properties.

\begin{assu}\label{as1}
  The interaction potential $V$ is assumed to be real-valued and
  reflection-symmetric, i.e., $V(x)=V(-x)$, with $V \in
  L^{3/2}(\R^3)$.  Moreover, the Schr\"odinger operator $-\nabla^2 +
  V(x)$ has a negative energy bound state.
\end{assu}

The $L^{3/2}$ assumption on $V$ guarantees relative form-bounded\-ness
with respect to the Laplacian.  The ground state energy of $-\nabla^2
+ V(x)$ will be denoted by $-\eb<0$, and its ground state wave
function by $\alpha_0$. It is unique up to a phase factor. We find it
convenient to normalize $\alpha_0$ such that
\begin{equation}\label{normalization}
  \int_{\R^3} |\widehat \alpha_0(q)|^2 \,\frac{dq}{(2\pi)^3} = 1\,,
\end{equation}
with $\widehat \alpha_0(q) = (2\pi)^{-3/2} \int_{\R^3} \alpha(x)
e^{-iq\cdot x} dx$ denoting the Fourier transform.

In the following, we are interested in the case $\mu = -\eb + h^2
\delta\mu$, which corresponds to the low density limit. We find it
convenient to absorb the constant $h^2 \delta \mu$ into the potential
$W$, i.e., we set $\mu = -\eb$ and write $W(x)$ instead of $W(x) -
\delta\mu$.

\subsection{The GP Functional}

Let $\psi\in H^1_{\rm per}(\R^3)$, the periodic functions in $H_{\rm
  loc}^1(\R^3)$. For $g\geq 0$, the GP functional is defined as
\begin{equation} \label{GPfunct} \E^{\rm GP}(\psi) = \int_{\calC}
  \left[ \tfrac 1 4 \left| \left( -i\nabla + 2
        A(x)\right)\psi(x)\right|^2 + W(x) |\psi(x)|^2 + g |\psi(x)|^4
  \right] dx\,.
\end{equation}
The coefficient $2$ in front of the vector potential $A$ is due to the
fact that $\psi$ describes pairs of particles, and the charge of a
pair is twice the particle charge. The factor $4$ in front of the
kinetic energy is twice the mass of a fermion pair. The coefficient
$g$ will be calculated below from BCS theory.

We denote the ground state energy of the GP functional by
\begin{equation}
  E^{\rm GP}(g) = \inf \left\{  \E(\psi)\, : \, \psi\in H^1_{\rm per}(\R^3)\right\}\,.
\end{equation}
It is not difficult to show that under our assumptions on $A$ and $W$,
there exists a corresponding minimizer, which satisfies a second order
differential equation known as the GP equation. Note that there is no
normalization constraint on $\psi$, the normalization is determined by
the chemical potential which is contained in $W(x)$.

\subsection{Main Results}\label{ss:results}

We define the energy $E^{\rm BCS}(\mu)$ as the infimum of $\F^{\rm
  BCS}$ over all admissible $\Gamma$, i.e.,
\begin{equation}\label{def:Fbcs}
  E^{\rm BCS}(\mu) = \inf_{\Gamma}\, \F^{\rm BCS}(\Gamma) \,.
\end{equation}
Recall that a state $\Gamma$ is admissible if $0\leq \Gamma\leq 1$ and $\Gamma$ is
periodic, i.e., it commutes with translations by $1$ in the three
coordinate directions.

Recall also that $\alpha_0$ denote the ground state of $-\nabla^2 + V(x)$,
normalized as in (\ref{normalization}).

\begin{Theorem}\label{thm:main}
  Let
  \begin{equation}\label{def:b3}
    g =      \int_{\R^3} |\widehat \alpha_0 (q)|^4 ( q^2 + \eb) \,\frac{dq}{(2\pi)^3} \,.
  \end{equation} 
  Under Assumptions~\ref{as0} and~\ref{as1} above, we have, for small
  $h$,
  \begin{equation}\label{enthm}
    E^{\rm BCS}(-\eb) =  h\left(E^{\rm GP}(g)  + e \right) \,,
  \end{equation}
  with $e$ satisfying the bounds $\const h \geq e \geq - \const
  h^{1/5}$.  Moreover, if $\Gamma$ is an approximate minimizer of
  $\F^{\rm BCS}$ at $\mu = -\eb$, in the sense that $\F^{\rm
    BCS}(\Gamma)\leq h ( E^{\rm GP}(g) + \epsilon)$ for some small
  $\epsilon > 0$, then the corresponding $\alpha$ can be decomposed as
  \begin{equation}\label{thm:dec}
    \alpha  = \frac h2 \big( \psi( x) \widehat\alpha_0(-ih\nabla) + 
\widehat \alpha_0(-ih\nabla) \psi(x)\big)  + \sigma
  \end{equation}
  with $\E^{\rm GP}(\psi) \leq E^{\rm GP}(g) + \epsilon + \const
  h^{1/5}$ and
  \begin{equation}\label{sigmab}
    \int_{\calC\times \R^3} |\sigma(x,y)|^2 \, dx\, dy  \leq  \const h^{3/5}\,.
  \end{equation}
\end{Theorem}

To appreciate the bound (\ref{sigmab}), note that the square of the
$L^2(\calC\times \R^3)$ norm of the first term on the right side of
(\ref{thm:dec}) is of the order $h^{-1}$, and hence is much larger
than the one of $\sigma$. To leading order in $h$, the
pair wave function $\alpha(x,y)$ is thus given by
\begin{equation}
\frac{ \psi(x) + \psi(y) }{2(2\pi)^{3/2} h^2} \alpha_0(h^{-1}(x-y)) \,,
\end{equation}
with $\psi$ a minimizer of the GP functional (\ref{GPfunct}).  This
agrees with 
\begin{equation}
\frac {\psi(\tfrac 12 (x+y))} {(2\pi)^{3/2} h^2}
\alpha_0(h^{-1}(x-y))
\end{equation}
to leading order in $h$, hence $\psi$ describes the center of mass
motion of pairs of fermion with are bound in the ground state of
$-h^2\nabla^2 + V(x/h)$.

Note that $g$ in (\ref{def:b3}) is strictly positive, even for purely
attractive interaction potentials $V$.  In the limit of a point
interaction \cite[Sect.~I.1]{albe} with scattering length $\as>0$ we
have
\begin{equation}
  \widehat \alpha_0(q) = \frac{ \sqrt{8\pi}\, \eb^{1/4}}{q^2 + \eb} 
\quad \text{and} \quad E_b = \frac 1{\as^2}\,,
\end{equation}
and hence $g = 2\pi \as$. Since the mass of the fermion pairs is $2$,
this corresponds to a scattering length of $2 \,\as$ for the pair
scattering \cite{melo}. The factor $2$ is an artifact of the BCS
approximation; an investigation of the actual four-body problem with
pseudo-potentials predicts a scattering length $\approx 0.6\, \as$
\cite{schl}.

By varying the external potential $W$, our bounds (\ref{enthm}) on the ground state
energy  can be used to obtain bounds on the particle
density as well. In particular, the number of particles per unit
volume, $N$, can be calculated by replacing $W(x)$ by $W(x) +
\delta\mu$ and taking the derivative of the energy with respect to
$\tfrac 12 h^2 \delta\mu$ at $\delta\mu = 0$.  To leading order in
$h$, the result is that
\begin{equation}
  N =  \frac 2 h \int_{\calC} |\psi^{\rm GP}(x)|^2 \, dx
\end{equation}
where $\psi^{\rm GP}$ is a minimizer of the GP functional
(\ref{GPfunct}). The average particle density, in microscopic variables, is
$\rho= h^3 N = 2 h^2 \int |\psi^{\rm GP}|^2$ and is thus of order
$h^2$. Hence our scaling limit corresponds indeed to low density.

In the translation invariant case, where $A(x)=0$ and $W(x) =
-\delta\mu < 0$ is constant, the GP minimizer is given by $|\psi^{\rm
  GP}(x)|^2 = \delta\mu/(2 g)$, and $E^{\rm GP}(g)= -
\delta\mu^2/(4g)$. In particular, the ground state energy per
particle, which is equal to $E^{\rm BCS}(\mu)/N + \tfrac 12 \mu $, is
given by
\begin{equation}
-\tfrac 12 \eb + \tfrac 12 g \rho + \text{higher order in $\rho$}
\end{equation}
for
small density $\rho$.

\subsection{Outline of the paper}

In the following Section~\ref{sec:semi} we shall state our main
semiclassical estimates. These are a crucial input to obtain the
bounds in Theorem~\ref{thm:main}. They are an extension to zero temperature
 of the analogous expressions at positive temperature obtained
in \cite[Sect.~2]{FHSS}.  An upper bound on $E^{\rm BCS}$ will be
derived in Section~\ref{sec:up}, using the variational
principle. Finally, Section~\ref{sec:low1} contains the lower
bound. In this final section also the structure of approximate
minimizers will be investigated. This leads to a definition of the
order parameter $\psi$. Our proof follows closely the proof of the
main theorem in \cite{FHSS}, but is partly simpler due to the fact
that we work at zero temperature.

Throughout the proofs, $C$ will denote various different constants. We
will sometimes be sloppy and use $C$ also for expressions that depend
only on some fixed, $h$-independent, quantities like $\eb$ or
$\|W\|_\infty$, for instance.

\section{Semiclassical Estimates}\label{sec:semi}

This section contains the semiclassical estimates needed in the proof
of Theorem~\ref{thm:main}. Let $\psi$ be a periodic function in
$H^2_{\rm loc}(\R^3)$. Pick a reflection-symmetric and real-valued
function $t$, with the property that
\begin{equation}\label{t:as1}
  \partial^\gamma t \in L^{6}(\R^3) 
\end{equation}
 and 
\begin{equation}\label{t:as2}
  \int_{\R^3} \frac{|\partial^\gamma t(q)|^2}{ 1+q^2}\, dq < \infty
\end{equation}
for all multi-indices $\gamma \in \{0,1\,\dots,4\}^3$. We shall later
choose $t(q) = 2 (q^2 + \eb) \widehat \alpha_0(q)$, but the results of
this section are valid for general functions $t$ satisfying
(\ref{t:as1}) and (\ref{t:as2}).

Let $\Delta$ denote the periodic operator
\begin{equation}\label{def:delta}
  \Delta = - \frac h 2\left(\psi(x) t(-ih\nabla) + t(-ih\nabla) \psi(x)\right) \,, 
\end{equation}
and let 
\begin{equation}\label{hdelta}
  H_\Delta = \left( \begin{array}{cc}  \left(-i h \nabla + h A(x) \right)^2 -\mu + h^2 W(x) 
& \Delta \\ \bar\Delta & - \left(i h \nabla + h A(x) \right)^2 +\mu - h^2 W(x) 
\end{array} \right)
\end{equation}
on $L^2(\R^3)\otimes \C^2$, with $A$ and $W$ satisfying
Assumption~\ref{as0}.  We shall also assume that $\mu<0$. In the
following, we will investigate the trace per unit volume of the
negative part of $H_\Delta$. Specifically, we are interested in the
effect of the off-diagonal term $\Delta$ in $H_\Delta$, in the
semiclassical regime of small $h$.

\begin{Theorem}\label{thm:scl}
  Let $[s]_- = \tfrac 12 \left( |s|-s \right)$ denote the negative part. 
  For $\mu<0$,  the diagonal entries of the $2\times 2$
  matrix-valued operator $[H_\Delta]_--[H_0]_-$ are
  locally trace class, and the sum of their traces per unit volume (which will be denoted by $\Trs$) 
  equals
  \begin{align}\nn
    \Tr_0\left( [H_\Delta]_- - [H_0]_-\right) & = - h^{-1} E_1 - h E_2 + O(h^{2}) \left(
      \|\psi\|^4_{H^1(\calC)} + \|\psi\|^2_{H^1(\calC)} \right) \\
    &\quad + O(h^3) \left( \|\psi\|_{H^1(\calC)}^6+
      \|\psi\|_{H^2(\calC)}^2\right)\,, \label{210}
  \end{align}
  where
  \begin{equation}\label{def:e1}
    E_1 =  -\frac {1} 2 \|\psi\|_2^2  
    \int_{\R^3} \frac{t(q)^2 }{q^2 - \mu} \,  \frac{dq}{(2\pi)^3} 
  \end{equation}
  and
  \begin{align}\nn
    E_2 &= - \frac {1} 8 \sum_{j,k=1}^3 \langle \partial_j
    \psi| \partial_k \psi\rangle \int_{\R^3} t(q)
    \left[ \partial_j \partial_k t\right]\!(q)\, \frac 1{q^2 - \mu} \,
    \frac{dq}{(2\pi)^3} \\ \nonumber & \quad + \left( \frac{ 1}8
     \left\| (\nabla + 2 i A) \psi\right\|_2^2 + \frac {1 }2 \langle \psi|
    W|\psi\rangle \right) 
    \int_{\R^3} \frac{t(q)^2}{(q^2 - \mu)^2}   \frac{dq}{(2\pi)^3}  
\\ & \quad  + \frac {1} 8 \|\psi\|_4^4
    \int_{\R^3} \frac{ t(q)^4}{(q^2 - \mu)^3} \,
    \frac{dq}{(2\pi)^3} \,. \label{def:e2}
  \end{align}
  The error terms in (\ref{210}) of order $h^2$ and $h^3$ depend on
  $t$ only via bounds on the expressions (\ref{t:as1}) and
  (\ref{t:as2}), and are uniform in $\mu$ for $\mu<0$ bounded away
  from zero.
\end{Theorem}

Here, we use the short-hand notation $\|\psi\|_p$ for the norm on
$L^p(\calC)$. Likewise, $\langle \,\cdot\, | \, \cdot\, \rangle$
denotes the inner product on $L^2(\calC)$.

In general, the operator $[ H_\Delta]_- - [H_0]_-$ is not trace class
under our assumptions on $t$ and $\psi$. Hence the trace in
(\ref{210}) has to be suitably understood as the sum of the traces of
the diagonal entries. This issue is further discussed in the next
section.

The proof of Theorem~\ref{thm:scl} is very similar to the proof of
Theorem~2 in \cite{FHSS}. In the following, we shall limit ourselves
to explaining the main differences.

\begin{proof}[Sketch of proof]
  Since $\mu < 0$ and $W(x)$ and $\psi(x)t(-ih\nabla)$ are bounded, both
  $H_\Delta$ and $H_0$ have, for small enough $h$, a gap around $0$ in
  the spectrum. Hence the projector onto the negative spectral
  subspace can be written via a contour integral as
\begin{equation}\label{the}
\theta(-H_\Delta) =  \frac 1{2\pi i}\int_{\ell}  \frac 1{z-H_\Delta} \, dz\,,
\end{equation}
where $\theta(t) = 1$ for $t\geq 1$ and $0$ otherwise, and $\ell$ is
the contour $\{r - i ,\, r\in (-\infty,0] \}\cup \{i r, \, r\in
[-1,1]\} \cup \{-r + i,\, r\in [0,\infty) \}$. The integral has to be
understood as a suitable weak limit over finite contours. Similarly,
one obtains that
 \begin{equation}\label{int:rep}
   [H_0]_- - [H_\Delta]_-  =  \frac 1{2\pi i}\int_{\ell} z \, 
\left( \frac 1{z-H_\Delta} - \frac 1{z-H_0} \right)\, dz\,.
\end{equation}
The remaining analysis proceeds as in \cite[Sect.~8]{FHSS} (compare
with Eq.~(8.11) there), and we shall not repeat it here. In \cite{FHSS}, the
factor $z$ on the right side of (\ref{int:rep}) is replaced by
$-\beta^{-1} \ln (1+e^{-\beta z})$ and the contour is around the whole
real axis. For $\beta\to \infty$, this reduces to (\ref{int:rep}),
given the gap in the spectrum around 0.
\end{proof}

Our second semiclassical estimate concerns the upper off-diagonal term
of the projection onto the negative spectral subspace of $H_\Delta$,
$\theta(-H_\Delta)$, which we denote by
$\alpha_\Delta$. We are interested in its $H^1$ norm. In general, we
define the $H^1$ norm of a periodic operator $O$ by
\begin{equation}\label{def:h1}
  \|O\|_{H^1}^2 = \Tr \left[ O^\dagger \left(1-h^2\nabla^2\right) O  \right] \,.
\end{equation}
In other words, $\|O\|_{H^1}^2 = \|O\|_2^2 + h^2 \|\nabla
O\|_2^2$. Note that this definition is not symmetric, i.e.,
$\|O\|_{H^1} \neq \|O^\dagger\|_{H^1}$ in general.

\begin{Theorem}\label{lem3} Let $\varphi(q) = \tfrac 12 t(q)/(q^2
  -\mu)$. Under the same assumptions as in Theorem~\ref{thm:scl}, we
  have
  \begin{equation}
    \left\| \alpha_\Delta - \tfrac h2 \left(\psi(x) \varphi(-ih\nabla) 
+ \varphi(-ih\nabla) \psi(x) \right) \right\|_{H^1}  
\leq C h^{3/2}  \left(  \|\psi\|_{H^2(\calC)} +  \|\psi\|^3_{H^1(\calC)}\right) \,.
  \end{equation}
\end{Theorem}

The proof follows again along the same lines as the proof of the
corresponding Theorem~3 in \cite{FHSS}, and we shall only sketch the
differences.

\begin{proof}[Sketch of proof]
With the aid of (\ref{the}) we can write
\begin{equation}\label{rep:alp}
  \alpha_\Delta = \frac 1{2\pi i} \int_\ell  \left[ \frac 1{z-H_\Delta}\right]_{12}dz\,,
\end{equation}
where $[\,\cdot \, ]_{ij}$ stands for the $ij$ element of an
operator-valued matrix, and where the integral has to be suitably
understood as a weak limit, similarly to
(\ref{int:rep}). Alternatively, one could integrate over $-\ell$,
since the identity operator has vanishing off-diagonal terms.

Using the resolvent identity and the definitions of $\Delta$ and
$\varphi$ we find that
\begin{equation}
  \alpha_\Delta  = \frac{h}2 \left( \psi(x) \varphi(-i h \nabla ) 
+ \varphi(-ih \nabla)  \psi(x)\right)   + \sum_{j=1}^3 \eta_j \,,
\end{equation}
where
\begin{equation}\label{def:eta1}
  \eta_1 =  \frac h{4\pi i}\int_\ell  \left(\frac 1{z-k_0}\left[\psi, k_0\right] \frac{t}{z^2-k_0^2} 
+ \frac{t}{z^2-k_0^2} [\psi,k_0]\frac 1{z+k_0} \right) \, dz \,,
\end{equation}
\begin{equation}
  \eta_2 =  \frac 1{2\pi i } \int_\ell \frac 1{z-k_0} \left( (k-k_0) \frac 1{z-k}\Delta 
+ \Delta\frac 1{z+k_0} (k_0 - k)\right) \frac 1{z+k} \, dz
\end{equation}
and
\begin{equation}
  \eta_3 = \frac 1{2\pi i}  \int_\ell  \frac 1{z-k}\Delta \frac 1{z+k} \Delta^\dagger 
\frac 1{z-k} \Delta  \left[ \frac 1{ z-H_\Delta } \right]_{22} dz \,.
\end{equation}
Here, $t$ is short for the operator $t(-ih\nabla)$, $k=(-ih\nabla + h
A(x))^2 -\mu + h^2 W(x)$ and $k_0 = -h^2\nabla^2 -\mu$.

Proceeding as in \cite[Section~9]{FHSS} one sees that
  \begin{equation}
    \|\eta_1\|_{H^1} \leq C h^{3/2} \|\psi\|_{H^2(\calC)}\,,
  \end{equation}
  \begin{equation}
    \|\eta_2\|_{H^1} \leq C h^{3/2} \|\psi\|_{H^1(\calC)}
  \end{equation}
and 
  \begin{equation}
    \|\eta_3\|_{H^1} \leq  C h^{3/2} \|\psi\|_{H^1(\calC)}^3 \,.
  \end{equation}
  In the terms investigated in \cite{FHSS}, there is an additional
  factor $(1+e^{\beta z})^{-1}$ in the integrand, and the contour
  contains the whole real axis. Similarly as in the proof of
  Theorem~\ref{thm:scl}, this reduces to our case as $\beta\to
  \infty$.
\end{proof}

\section{Proof of Theorem~\ref{thm:main}: Upper Bound}\label{sec:up}

Recall that $\alpha_0$ denotes the unique ground state of $-\nabla^2 +
V(x)$, normalized as in (\ref{normalization}). It satisfies
$\alpha_0(x) = \alpha_0(-x)$, and we can take it to be real. In the
following, we let $t$ denote the Fourier transform of $2 (-\nabla^2 +
\eb) \alpha_0 = -2 V \alpha_0$, i.e.,
\begin{equation}\label{deft}
  t(q) = - 2 (2\pi)^{-3/2} \int_{\R^3} V(x) \alpha_0(x) e^{-iq\cdot x} dx 
  = 2 (q^2 + \eb) \widehat\alpha_0(q) \,.
\end{equation}
It satisfies all the assumptions in the previous section. In
particular, (\ref{t:as1}) and (\ref{t:as2}) hold for all $\gamma \in
\N_0^3$. This can be shown, for instance, in the same way as in
\cite[Sect.~4]{FHSS}. The method there also implies that
$\sqrt{|V(x)|}\alpha_0(x) e^{\kappa|x|} \in L^2(\R^3)$ for $\kappa<
\eb^{1/2}$, and that $\int_{\R^3} ( |x^\gamma \nabla\alpha_0(x)|^2 +
|x^\gamma \alpha_0(x)|^2 ) dx < \infty$ for all $\gamma \in
\N_0^3$. Some of these properties will be used later on.

As a trial state, we use
\begin{equation}\label{def:gammad2}
  \Gamma_\Delta = \left( \begin{array}{cc} \gamma_\Delta & \alpha_\Delta  
      \\ \bar\alpha_\Delta & 1-\bar\gamma_\Delta \end{array} \right) = \theta(-H_\Delta)
\end{equation}
where $H_\Delta$ is given in (\ref{hdelta}) with $\Delta$ as in
(\ref{def:delta}) and $\mu = -\eb$.  For $t$, we choose (\ref{deft}),
which is reflection symmetric and can be taken to be real.

We have
\begin{equation}\label{tci2} [H_0]_- - [H_\Delta]_- = H_\Delta
  \Gamma_\Delta - H_0 \Gamma_0 = \left( \begin{array}{cc} k
      \gamma_\Delta + \Delta \bar\alpha_\Delta & k \alpha_\Delta +
      \Delta ( 1- \bar\gamma_\Delta) \\ \bar\Delta \gamma_\Delta +
      \bar k \bar\alpha_\Delta & \bar k \bar\gamma_\Delta  +
      \bar\Delta \alpha_\Delta \end{array} \right)
\end{equation}
where $k$ denotes the upper left  entry of $H_\Delta$ (and $H_0$). From
(\ref{def:bcs}) and (\ref{tci2}) we conclude that 
\begin{align}\nonumber
\F^{\rm
  BCS}(\Gamma_\Delta) =  &-\frac 1{2}\, \Trs\left( [H_\Delta]_- -[H_0]_-\right) \\ \nonumber
  &- h^{-4} \int_{\calC\times \R^3} V(\tfrac {x-y}h)
\left|\tfrac 12 (\psi(x)+\psi(y))\alpha_0(\tfrac{x-y}h)\right|^2\, 
\frac{dx\,dy}{(2\pi)^3} \\
  & + \int_{\calC\times\R^3} V(\tfrac{x-y}h)\left|
    \frac{\psi(x)+\psi(y)}{2h^2(2\pi)^{3/2}}
    \alpha_0(\tfrac{x-y}h)-
    \alpha_\Delta(x,y)\right|^2\,{dx\,dy}\,, \label{fund:up}
\end{align}
where $\Trs$ stands for the sum of the traces per unit volume of the diagonal
entries of the $2\times 2$ matrix-valued operator. In general, the
operator $[H_0]_--[H_\Delta]_-$ is not trace class if $\Delta$ is not,
as can be seen from (\ref{tci2}). In the evaluation of $\F^{\rm
  BCS}(\Gamma_\Delta)$ only the diagonal terms of (\ref{tci2}) enter,
however.

The first term on the right side of (\ref{fund:up}) was calculated in
Theorem~\ref{thm:scl} above. Note that, for our choice of $t$, the
integral in the second term in (\ref{def:e2}) is equal to $4$, and
the integral in the third term is $16\, g$.

As in \cite[Sect.~5]{FHSS} we can rewrite the second term on the right
side of (\ref{fund:up}) as
\begin{align}\nonumber
  & - h^{-4} \int_{\R^3\times \calC} V(h^{-1}(x-y))\left|\tfrac 12
    (\psi(x)+\psi(y))\alpha_0(h^{-1}(x-y))\right|^2\, dx\,dy \\
  & = \frac {1 }{16 h} \sum_{p\in (2\pi\Z)^3} |\widehat \psi(p)|^2
  \int_{\R^3} \frac{t(q)}{q^2+\eb} \, \left(2 t(q) + t(q-hp) +
    t(q+hp)\right) \, dq\,. \label{term2}
\end{align}
Using the Taylor expansion
\begin{align}\nn
  & 2t(q) + t(q-hp) + t(q+hp) \\ & = 4 t(q) + {h^2} \left[ (p\cdot
    \nabla)^2 t\right]\! (q) + \frac{h^4}{6} \int_{-1}^1 \left[
    (p\cdot \nabla)^4t\right]\!(q+shp) (1-|s|)^3 \, ds
\end{align}
we see that (\ref{term2}) equals
\begin{equation}\label{3t}
  \frac {\| \psi\|_2^2 }{4h} \int_{\R^3} \frac{t(q)^2}{q^2+\eb}  \, dq 
+ \frac {h  }{16} \sum_{i,j=1}^3 \langle \partial_i
  \psi| \partial_j\psi\rangle \int_{\R^3} t(q)
  \left[ \partial_i\partial_j t\right]\! (q) \, \frac 1{q^2+\eb}
  \, dq + O(h^{3})\,,
\end{equation}
where the error term is bounded by
\begin{equation}
  C h^{3} \|\psi\|_{H^2}^2 \int_{\R^3} |V(x)| |\alpha_0(x)|^2 |x|^4 \, dx\,.
\end{equation}
Note that the first term in (\ref{3t}) cancels the contribution of
$h^{-1} E_1$ in (\ref{def:e1}) to the trace $\tfrac 1{2}\, \Tr(
[H_\Delta]_- -[H_0]_-)$.

It remains to investigate the last term in (\ref{fund:up}). Since $V$
is relatively bounded with respect to the Laplacian, we can bound the
term by an appropriate $H^1$ norm. Recall the definition of the $H^1$
norm of a periodic operator in (\ref{def:h1}).  For general periodic
operators $O$, we have the bound
\begin{equation}
  \left| \int_{\calC\times\R^3} V(\tfrac{x-y}h)\left| O(x,y)
    \right|^2\,dx\,dy \right| \leq  \left\|
    (1-\nabla^2)^{-1/2} V(\,\cdot\,) (1- \nabla^2)^{-1/2} \right\| \,
  \|O\|_{H^1}^2\,.
\end{equation}
The operator of relevance here is given by
\begin{equation}\label{relo}
  O  = \alpha_\Delta  - \tfrac h 2 \left ( \psi(x) \widehat\alpha_0(-ih\nabla) 
+ \widehat\alpha_0(-ih\nabla) \psi(x)\right) \,.
\end{equation} 
Note that, for our choice of $t$, we have $\varphi(q) = \tfrac 12
t(q)/(q^2 + \eb) = \widehat\alpha_0(q)$. Hence Theorem~\ref{lem3}
implies that the $H^1$ norm of (\ref{relo}) is bounded by $C h^{3/2} (
\|\psi\|_{H^2} +\|\psi\|_{H^1}^3)$.

For $\psi$, we shall take a minimizer of the GP functional
(\ref{GPfunct}). Under Assumption~\ref{as0} on $W$ and $A$, it is
easily seen to be in $H^2$. Collecting all the terms, we see that for
this choice of $\psi$ we have
\begin{equation}
  E^{\rm BCS}(-\eb) \leq \F^{\rm BCS}(\Gamma_\Delta)  \leq  h \left( E^{\rm GP} + C h \right)
\end{equation}
for small $h$.  This completes the proof of the upper bound.

\section{Proof of Theorem~\ref{thm:main}: Lower Bound}\label{sec:low1}

Our proof of the lower bound on $E^{\rm BCS}(-\eb)$ in
Theorem~\ref{thm:main} consists of two main parts.  The goal of this
first part is to show the following.

\begin{prop}\label{new:prop}
  Let $\Gamma$ be a state satisfying $\F^{\rm BCS}(\Gamma)\leq 0$, and
  let $\alpha$ denote its off-diagonal element. Then there exists a
  periodic function $\psi$, with $H^1(\calC)$ norm bounded independently
  of $h$, such that
  \begin{equation}\label{defpsinew}
    \alpha = \tfrac h 2 \big( \psi(x) \widehat \alpha_0(-ih\nabla) 
+ \widehat \alpha_0(-ih\nabla) \psi(x)\big)  + \xi
  \end{equation}
  with $\|\xi\|_{H^1} \leq O(h^{1/2})$ for small $h$.
\end{prop}

Recall the definition (\ref{def:h1}) for the $H^1$ norm of a periodic
operator. The bound $\|\xi\|_{H^1} \leq O(h^{1/2})$ has to be compared
with the $H^1$ norm of the first part of (\ref{defpsinew}), which is
$O(h^{-1/2})$ (for fixed $\psi\neq 0$.)

\begin{proof}
  Let $K^{A,W}$ denote the operator
  \begin{equation}\label{def:ktaw}
    K^{A,W} = \left(-i h \nabla + h A(x) \right)^2 +\eb + h^2 W(x)\,,
  \end{equation}
  and $\lambda^{A,W} = \infspec K^{A,W} \geq \eb - O(h^2)$.  For all
  admissible states $\Gamma$, we have
  \begin{equation}
   0\leq \alpha \bar \alpha \leq \gamma ( 1 - \gamma) \leq \gamma 
  \end{equation}
  and hence
  \begin{align}\nn
    \Tr K^{A,W} \gamma & \geq \Tr K^{A,W} \left( \alpha \bar \alpha +
      \gamma^2\right) \geq \Tr K^{A,W} \alpha \bar \alpha +
    \lambda^{A,W} \Tr \gamma^2 \\ &\geq \Tr K^{A,W} \alpha \bar \alpha
    + \lambda^{A,W} \Tr (\alpha \bar \alpha)^2\,.
  \end{align}
  In particular,
  \begin{equation}\label{eq:step1}
    \F^{\rm BCS}(\Gamma) \geq   \lambda^{A,W} \Tr ( \alpha\bar\alpha )^2 
+ \int_{\calC} \langle \alpha(\,\cdot\,,y) | K^{A,W} + V(h^{-1}(\,\cdot\,-y)) 
| \alpha(\,\cdot\,,y)\rangle \, {dy}\,.
  \end{equation}
  Here $K^{A,W}$ acts on the $x$ variable of $\alpha(x,y)$, and
  $\langle\, \cdot \, | \, \cdot \,\rangle$ denotes the standard inner
  product on $L^2(\R^3)$.

  By definition, the operator $ K^{0,0} + V(h^{-1}(\,\cdot\,-y))$ on
  $L^2(\R^3)$ has a unique ground state
  $h^{-3/2}\alpha_0(h^{-1}(x-y))$, with ground state energy zero, and
  a gap above. To utilize this fact, it will be convenient to replace
  $K^{A,W}$ by $K^{0,0}$ in (\ref{eq:step1}). We claim that
  \begin{equation}\label{new:claim}
    K^{A,W} + V(h^{-1}(\,\cdot\,-y)) \geq \frac 12 \left( K^{0,0} 
+ V(h^{-1}(\,\cdot\, -y) ) \right) -  h^2 \left( \|W\|_\infty + \|A\|_\infty^2\right)  \,.
  \end{equation}
  This follows immediately from the fact that
  \begin{align}\nn
    K^{A,W} + V (h^{-1}(\,\cdot\,-y)) & = \frac 12 \left( K^{0,0} + V
      (h^{-1}(\,\cdot\,-y)) \right) \\ & \quad + \frac 12 \left( K^{2
        A,0} + V (h^{-1}(\,\cdot\,-y)) \right) - h^2 A(x)^2 + h^2 W(x)
  \end{align}
  and that $K^{2A,0}+ V(h^{-1}(\,\cdot\,-y)) \geq 0$, by the
  diamagnetic inequality.

  For any state $\Gamma$ with $\F^{\rm BCS}(\Gamma)\leq 0$ we conclude
  from (\ref{eq:step1}) and (\ref{new:claim}) that
  \begin{equation}\label{eq1}
    \lambda^{A,W}  \Tr ( \alpha\bar\alpha )^2 
+ \frac 12 \int_{\calC} \langle \alpha(\,\cdot\,,y) | K^{0,0} 
+ V(h^{-1}(\,\cdot\,-y)) | \alpha(\,\cdot\,,y)\rangle \,  {dy}\leq C h^2 \|\alpha\|_2^2  \,.
  \end{equation}
  To show that this inequality implies (\ref{defpsinew}), we shall
  proceed as in \cite[Sect.~6]{FHSS}.

  Define $\psi$ to be the periodic function
  \begin{equation}
    \psi(y) =  \frac 1{(2\pi)^{3/2}h}  \int_{\R^3} \alpha_0(h^{-1}(x-y)) \alpha(x,y) \,dx\,.
  \end{equation}
  If we write
  \begin{equation}\label{dec1}
    \alpha(x,y)= \frac{1}{(2\pi)^{3/2}h^2} \alpha_0(h^{-1}(x-y))\psi(y)  + \xi_0(x,y)
  \end{equation}
  the gap in the spectrum of $K^{0,0} + V(h^{-1}(\,\cdot\,-y))$ above
  zero, together with (\ref{eq1}) and the normalization
  (\ref{normalization}), yields the bound $\|\xi_0\|_2 \leq
  O(h)\|\alpha\|_2$. We can also symmetrize and write
  \begin{equation}
    \alpha(x,y) = \frac{ \psi( x)+\psi(y)}{2(2\pi)^{3/2}h^2} \alpha_0(h^{-1}(x-y)) + \xi(x,y) \,,
  \end{equation}
  again with $\|\xi\|_2\leq O(h)\|\alpha\|_2$. In order to complete
  the proof of (\ref{defpsinew}), we need to show that
  $\|\psi\|_{H^1}$ is bounded independently of $h$, and that the $H^1$
  norm of $\xi$ is bounded by $O(h^{1/2})$.

  An application of Schwarz's inequality yields
  \begin{equation}
    \int_\calC | \psi(x)|^2\, dx \leq  h  \|\alpha\|_2^2  
\leq   \int_\calC |\psi(x)|^2 dx  + h \|\xi_0\|_2^2\,.
  \end{equation}
  Since $\|\xi_0\|_2\leq O(h) \|\alpha\|_2$, this implies that
  \begin{equation}\label{schw1}
    \|\alpha\|_2^2 \leq (1 + O(h^2)) \frac 1 h  \int_\calC |\psi(x)|^2 dx \,.
  \end{equation}
  Again by using Schwarz's inequality,
  \begin{equation}\label{schw2}
    \int_\calC |\nabla \psi(x)|^2\, dx \leq  \frac 1 h  \int_{\R^3\times \calC}  
\left| \left(\nabla_x + \nabla_y\right) \alpha(x,y)\right|^2 \, dx\,dy  \,.
  \end{equation}
  The latter expression can be bounded as
  \begin{equation}\label{lem:eq:com}
    \int_{\R^3\times \calC} \left| \left(\nabla_x +
        \nabla_y\right) \alpha(x,y)\right|^2 \, dx\,dy  \leq \frac 4{h^2} 
    \int_{\calC} \langle \alpha(\,\cdot\,,y) | K^{0,0} +
    V(\tfrac {\,\cdot\,-y}h) | \alpha(\,\cdot\,,y)\rangle \,
    dy \,.
  \end{equation}
  To see this, expand $\alpha(x,y)$ in a Fourier series
  \begin{equation}
    \alpha(x,y) = \sum_{ p \in (2\pi \Z)^3} e^{i p\cdot (x+y)/2} \widetilde \alpha_p(x-y) \,.
  \end{equation}
  Using that $\widetilde \alpha_p(x) = \widetilde \alpha_p(-x)$ for
  all $p \in (2\pi \Z)^3$ we see that (\ref{lem:eq:com}) is equivalent
  to
  \begin{equation}
    K^{\frac 1 2 p,0}  + K^{-\frac 1 2 p,0} +2 \, V(x/h) \geq \frac 12 h^2 p^2 \,.
  \end{equation}
  This holds, in fact, for all $p\in \R^3$ since the left side is
  equal to $2 K^{0,0} + \frac 12 h^2 p^2 + 2\eb$.

  By combining (\ref{lem:eq:com}) with (\ref{schw1}), (\ref{schw2})
  and (\ref{eq1}) we see that $\|\nabla\psi\|_2$ is bounded by a
  constant times $\|\psi\|_2$. To conclude the uniform upper bound on
  the $H^1$ norm of $\psi$, it thus suffices to give a bound on the
  $L^2$ norm. To do this, we have to utilize the first term on the
  left side of Eq.~(\ref{eq1}).

  Eq.~(\ref{dec1}) states that $\alpha$ can be decomposed as $\alpha =
  h \alpha_0 \psi + \xi_0$, where $\alpha_0$ is short for the operator
  $\widehat\alpha_0(-ih \nabla)$. The following lemma was proved in
  \cite[Lemma~6]{FHSS}.  It gives a lower bound on $ ( \Tr
  (\alpha\bar\alpha)^2 )^{1/4}$, the $4$-norm of $\alpha$. This bound holds
  under appropriate decay and smoothness assumptions on $\alpha_0$
  which are satisfied in our case. (See the discussion at the
  beginning of Section~\ref{sec:up}.)

\begin{Lemma}\label{lem:a4}
  For some $0<C<\infty$ we have
  \begin{align}\nn
    \|\alpha\|_4 & \geq \left[ h \int_\calC |\psi(x)|^4 \, dx
      \int_{\R_3} \widehat\alpha_0(q)^4 \, \frac{dq}{(2\pi)^3} - C
      h^{2} \|\psi\|_{H^1(\calC)}^4 \right]_+^{1/4} \\ & \quad - C
    h^{1/4} \|\psi\|_2^{1/2} \left ( 1 + C h^{1/4} \|\psi\|_4
    \right)^{1/2}\,,\label{eq:lem:a4}
  \end{align}
  where $[\,\cdot\,]_+ = \max\{ 0, \,\cdot\,\}$ denotes the positive
  part.
\end{Lemma}

The fact that $\|\nabla\psi\|_2\leq C \|\psi\|_2$ also implies that
$\|\psi\|_4 \leq C \|\psi\|_2$ via Sobolev's inequality for functions
on the torus. If we use also that $\|\psi\|_4 \geq \|\psi\|_2$ we
conclude from (\ref{eq:lem:a4}) that $\|\alpha\|_4 \geq C h^{1/4}
(\|\psi\|_2 - C \|\psi\|_2^{1/2})$ for $h$ small enough. In
combination with (\ref{eq1}) and (\ref{schw1}) this implies that
$\|\psi\|_2 \leq C$. This shows that the $H^1$ norm of $\psi$ is
indeed uniformly bounded.

It follows that $\|\xi\|_2 \leq O(h^{1/2})$. To conclude the proof of
(\ref{defpsinew}), we need to show that also $\|\xi\|_{H^1} \leq
O(h^{1/2})$. We can write
\begin{equation}
  \xi(x,y) = \xi_0(x,y) + \frac{\psi(x) - \psi(y)}{2(2\pi)^{3/2}h^2} \alpha_0(h^{-1}(x-y)) \,.
\end{equation}
From the definition (\ref{dec1}) it follows easily that $\|
\xi_0\|_{H^1}\leq O(h^{1/2})$, using that $-h^2\nabla^2$ is relatively
bounded with respect to $K^{0,0} + V(h^{-1}(\,\cdot\,-y)) + 1$. If we
use the boundedness of the $H^1$ norm of $\psi$ and, moreover,
\begin{align}\nn
  &h^{-3} \int_{\calC\times\R^3} |\psi(x)-\psi(y)|^2 |\nabla
  \alpha_0(h^{-1}(x-y))|^2\,dx \, dy \\ & = 4 \sum_{p\in (2\pi\Z)^3}
  |\widehat \psi(p)|^2 \int_{\R^3} |\nabla \alpha_0(x)|^2 \sin^2\left(
    \tfrac 12 h p\cdot x \right) \, dx \leq O(h^{2})
\end{align}
(since $\int |\nabla\alpha_0|^2 |x|^2dx$ is finite), the bound on the
$H^1$ norm of $\xi$ readily follows. This completes the proof of
Proposition~\ref{new:prop}.
\end{proof}

Given Proposition~\ref{new:prop}, the proof of the lower bound on the ground state energy is very
similar to the corresponding one in \cite[Sect.~7]{FHSS}.  Let
$\Gamma$ be a state with $\F^{\rm BCS}(\Gamma) \leq 0$, and let $\psi$
be the function defined by the decomposition (\ref{defpsinew}). In
order to be able to apply Theorems~\ref{thm:scl} and \ref{lem3}, we
have to make sure that $\psi$ is in $H^2$. For this purpose, we pick
some $\epsilon>0$ with $h<\epsilon<1$ and define $\psi_<$ via its
Fourier coefficients
\begin{equation}
  \widehat
  \psi_<(p) = \widehat \psi(p) \theta(\epsilon h^{-1} - |p|) \,.
\end{equation}
The function $\psi_<$ is thus smooth, and $\|\psi_<\|_{H^2}\leq C
\epsilon h^{-1}$ since $\psi$ is bounded in $H^1$.

Let also $\psi_> = \psi - \psi_<$. Since $\psi$ is bounded in $H^1$,
the $L^2(\calC)$ norm of $\psi_>$ is bounded by $O(h \epsilon^{-1})$.
We absorb the part $\frac
12(\psi_>(x)+\psi_>(y))\alpha_0(h^{-1}(x-y))$ into $\xi$, and write
\begin{equation}\label{defpsi}
  \alpha(x,y) =  \frac{ \psi_<(x)+\psi_<(y)} { 2 (2\pi)^{3/2}h^2} \alpha_0(h^{-1}(x-y)) +\sigma(x,y)
\end{equation}
where
\begin{equation}\label{def:sigma}
  \sigma (x,y) = \xi(x,y) + \frac{ \psi_>(x)+\psi_>(y)}{2(2\pi)^{3/2}h^2} \alpha_0(h^{-1}(x-y)) \,.
\end{equation}
Proposition~\ref{new:prop} shows that $\|\xi\|_{H^1}\leq O(h^{1/2})$.
From the bound $\|\psi_>\|_2 \leq O(h\epsilon^{-1})$ it thus follows
that $\|\sigma\|_2 \leq O(h^{1/2}\epsilon^{-1})$. We cannot conclude
the same bound for the $H^1$ norm of $\sigma$, however.

As in (\ref{def:delta}), let $\Delta$ denote the operator $\Delta =
-\frac 12 (\psi_<(x) t(-ih\nabla) + t(-ih\nabla) \psi_<(x))$. The
function $t$ is given in (\ref{deft}), as in  the
previous section. Let $H_\Delta$ be the corresponding Hamiltonian
defined in (\ref{hdelta}). We can write
\begin{align}\nonumber
  \F^{\rm BCS}(\Gamma) & = -\frac 1 2\, \Trs\left( [ H_\Delta]_- -
    [H_0]_- \right) \\ \nonumber & \quad - \frac{1}{4h^4}
  \int_{\calC\times \R^3} V(\tfrac{x-y}h) \left|
    \psi_<(x)+\psi_<(y)\right|^2 |\alpha_0(\tfrac{x-y}h)|^2\,
  \frac{dx\,dy}{(2\pi)^3}\\ & \quad + \frac 12\, \Trs  H_\Delta
    (\Gamma-\Gamma_\Delta)  + \int_{\calC\times \R^3}
  V(h^{-1}(x-y))|\sigma(x,y)|^2\, {dx\,dy}\,, \label{off}
\end{align}
where $\Trs$ denotes again the sum of the trace per unit
volume of the diagonal entries, as in (\ref{fund:up}).

The terms in the first two lines on the right side of (\ref{off}) have
already been calculated. The first term is estimated in
Theorem~\ref{thm:scl}, and a bound on the second term was derived in
Section~\ref{sec:up} on the upper bound.  Using the fact that the
$H^1$ norm of $\psi_<$ is uniformly bounded, as well as
$\|\psi_<\|_{H^2} \leq C \epsilon/h$, we obtain the lower bound
\begin{align}\nonumber
  \F^{\rm BCS}(\Gamma) &\geq h \left( \E^{\rm GP}(\psi_<) - C (
    h+\epsilon^2) \right) \\ & \quad + \frac 12\, \Trs  H_\Delta
    (\Gamma-\Gamma_\Delta)  + \int_{\calC\times \R^3}
  V(h^{-1}(x-y))|\sigma(x,y)|^2\, {dx\,dy}\,. \label{lb2:s1}
\end{align}
It remains to show that the terms in the last line of (\ref{lb2:s1})
are negligible, i.e., of higher order than $h$, for an appropriate
choice of $\epsilon \ll 1$. We shall use the following lemma, whose proof is inspired by \cite[Lemma~1]{HLS}.

\begin{Lemma}
For all $0\leq \Gamma\leq 1$ with $(-\nabla^2+1)\gamma$ trace class, we have
\begin{equation}\label{eqs2}
  \Trs  H_\Delta (\Gamma-\Gamma_\Delta)   \geq   
\Tr \left( \Gamma-\Gamma_\Delta\right) |H_\Delta|\left( \Gamma-\Gamma_\Delta\right)  \,.
\end{equation}
\end{Lemma}

\begin{proof} 
Recall that $\Gamma_\Delta$ is the projection onto the negative spectral
subspace of $H_\Delta$. Moreover, for any operator $A$, $\Trs A = \Tr \Gamma_0 A \Gamma_0 + \Tr (1-\Gamma_0) A (1-\Gamma_0)$. A simple calculation shows that
\begin{align}\nonumber
& \Gamma_0 H_\Delta(\Gamma-\Gamma_\Delta) \Gamma_0  + (1-\Gamma_0) H_\Delta(\Gamma-\Gamma_\Delta) (1- \Gamma_0)  \\ & = H_\Delta^+ \Gamma ( 1- \Gamma_\Delta)  + H_\Delta^- (1-\Gamma) \Gamma_\Delta - E_1 - E_2 \label{eqw}
\end{align}
with $H_\Delta^\pm$ denoting the positive and negative parts of $H_\Delta$, respectively, 
\begin{equation}
E_1 = (\Gamma_0 - \Gamma_\Delta ) H_\Delta (\Gamma-\Gamma_\Delta) (1 - 2 \Gamma_0)
\end{equation}
and 
\begin{equation}
E_2 =  |H_\Delta| (\Gamma-\Gamma_\Delta) (\Gamma_0 - \Gamma_\Delta )\,.
\end{equation}
It is easy to see that $(\Gamma_0-\Gamma_\Delta ) |H_0|^{1/2}$ and $|H_0|^{1/2}(\Gamma-\Gamma_\Delta)$ are Hilbert-Schmidt; since $\Delta$ is bounded, also $(\Gamma_0-\Gamma_\Delta ) |H_\Delta|^{1/2}$ and $|H_\Delta|^{1/2}(\Gamma-\Gamma_\Delta)$ are Hilbert-Schmidt. Hence $E_1$ is trace class and, by cyclicity, its trace is equal to the one of 
\begin{align}\nonumber
\widetilde E_1 & = \sqrt{H_\Delta^+} (\Gamma-\Gamma_\Delta) (1-2\Gamma_0) (\Gamma_0-\Gamma_\Delta) \sqrt{H_\Delta^+} \\ \nonumber & \quad  - \sqrt{H_\Delta^-} (\Gamma-\Gamma_\Delta) (1-2\Gamma_0) (\Gamma_0 - \Gamma_\Delta ) \sqrt{H_\Delta^-} \\ & = - \sqrt{H_\Delta^+} (\Gamma-\Gamma_\Delta)( \Gamma_0 -\Gamma_\Delta) \sqrt{H_\Delta^+} - \sqrt{H_\Delta^-} (\Gamma-\Gamma_\Delta)(\Gamma_0-\Gamma_\Delta) \sqrt{H_\Delta^-}\,.
\end{align}

Via the Floquet decomposition, $H_\Delta$ can be written as a direct integral of operators on $L^2(\calC)$ each of which has discrete spectrum. Since we know, a priori, that (\ref{eqw}) is trace class, we can evaluate the trace in the basis given by $H_\Delta$. With this understanding of the trace, we have $\Tr [\widetilde E_1 + E_2] =0$, and thus  
\begin{align}\nn
 \Trs H_\Delta(\Gamma-\Gamma_\Delta)  & = \Tr \left[ H_\Delta^+ \Gamma ( 1- \Gamma_\Delta)  + H_\Delta^- (1-\Gamma) \Gamma_\Delta \right] \\ & =\Tr \left[ \sqrt{H_\Delta^+} \Gamma \sqrt{H_\Delta^+} + \sqrt{H_\Delta^-} ( 1-\Gamma) \sqrt{H_\Delta^-}\right]\,.
\end{align}
The operators on the last line are positive, hence they are trace class.  Estimating $\Gamma\geq \Gamma^2$ and $1-\Gamma \geq (1-\Gamma)^2$, respectively, gives the desired bound
\begin{align}\nn
 \Trs  H_\Delta(\Gamma-\Gamma_\Delta)  & \geq \Tr \left[ \sqrt{H_\Delta^+} \Gamma^2 \sqrt{H_\Delta^+} + \sqrt{H_\Delta^-} ( 1-\Gamma)^2 \sqrt{H_\Delta^-}\right] \\ & = \Tr \left[ \sqrt{H_\Delta^+} (\Gamma-\Gamma_\Delta)^2 \sqrt{H_\Delta^+} + \sqrt{H_\Delta^-} ( \Gamma-\Gamma_\Delta)^2 \sqrt{H_\Delta^-}\right]\,,
\end{align}
which agrees with the right side of (\ref{eqs2}). 
\end{proof}

An application of Schwarz's inequality yields
\begin{equation}\label{hhd}
  H_\Delta^2 \geq (1-\eta) H_0^2 - \eta^{-1}  \|\Delta\|_\infty^2 
\end{equation}
for any $\eta>0$. Schwarz's inequality can also be used to obtain a
lower bound on $H_0^2$. For any $0<\delta<1$,
\begin{multline}
  \left[ \left( -ih\nabla + h A(x) \right)^2 + \eb + h^2 W(x)
  \right]^2 \geq (1-\delta)^2 \left[ -h^2\nabla^2 + \eb \right]^2 \\ -
  \frac 1 \delta \left[ \left( -ih^2\nabla \cdot A(x) - i h^2
      A(x)\cdot \nabla\right)^2 + h^4 \left( W(x) + A(x)^2\right)^2
  \right] \,.
\end{multline}
We can further bound
\begin{align}\nn
  & \left( -ih^2\nabla \cdot A(x) - i h^2 A(x)\cdot \nabla\right)^2 \\
  \nn & = \left( -2ih^2 \nabla\cdot A(x) + i h^2 \dive A(x) \right)
  \left( -2 i h^2 A(x)\cdot \nabla - i h^2 \dive A(x) \right) \\ & \leq
  8 h^4 \nabla \cdot A(x) A(x) \cdot \nabla + 2 h^4 \left( \dive A(x)
  \right)^2 \,.
\end{align}
Since $A$ is $C^1$ by assumption, this is bounded from above by $C h^4
(- \nabla^2 +1)$. Choosing $\delta = O(h)$, we thus conclude that
\begin{equation}
  H_0^2 \geq  (1-O(h)) [-h^2\nabla^2 +\eb]^2\otimes \id_{\C^2} \,.
\end{equation}

The operator monotonicity of the square root implies that
\begin{equation}
  K^{0,0}\otimes \id_{\C^2} \leq  (1-\eta-O(h))^{-1/2}
  \sqrt{H_\Delta^2 + \eta^{-1}\|\Delta\|_\infty^2}  \,.
\end{equation}
Using again (\ref{hhd}) and the fact that $H_0^2 \geq
(\lambda^{A,W})^2 \geq O(1)$, the choice $\eta = O(\|\Delta\|_\infty)$
gives
\begin{equation}
  |H_\Delta| \geq (1-O(h+\|\Delta\|_\infty))  K^{0,0} \otimes \id_{\C^2}\,.
\end{equation}
In particular, we infer from (\ref{eqs2}) that
\begin{equation}\label{eqs3}
  \tfrac 12  \,  \Tr  H_\Delta (\Gamma-\Gamma_\Delta)  \geq    
(1-O(h+\|\Delta\|_\infty)) \, \Tr K^{0,0} 
(\alpha-\alpha_\Delta)(\bar\alpha-\bar\alpha_\Delta)\,,
\end{equation}
where $\alpha_\Delta$ denotes again the upper off-diagonal entry of
$\Gamma_\Delta$.  From the definition of $\Delta$, we see that
\begin{equation}
  \|\Delta\|_\infty \leq  h \|\psi_<\|_\infty \| t \|_\infty \,. 
\end{equation}
Moreover, since the Fourier transform of $\psi_<$ is supported in the
ball $|p|\leq \epsilon/h$,
\begin{equation}
  \|\psi_<\|_\infty \leq \sum_p |\widehat\psi_<(p)|  \leq
  \|\psi_<\|_{H^1(\calC)} \left( \sum_{|p|\leq \epsilon h^{-1}} \frac
    1{1+p^2} \right)^{1/2} \leq C  \sqrt{\epsilon/h}\,,
\end{equation}
and hence $\|\Delta\|_\infty \leq O(\epsilon^{1/2} h^{1/2})$. 

Recall the decomposition (\ref{defpsi}) of $\alpha$. We decompose
$\alpha_\Delta$ in a similar way, and define $\phi$ by
\begin{equation}
  \alpha_\Delta = \tfrac h2 \left( \psi_<(x) \widehat\alpha_0(-ih\nabla) 
+ \widehat \alpha_0(-ih\nabla) \psi_<(x)\right) +\phi \,.
\end{equation}
In particular, we have
\begin{equation}
  \alpha - \alpha_\Delta = \sigma  - \phi \,.
\end{equation}
Since $\|\psi_<\|_{H^2}\leq O(\epsilon/h)$, Theorem~\ref{lem3} implies
that $\|\phi\|_{H^1} \leq O(\epsilon h^{1/2})$.  From the positivity
of $K^{0,0}$ we conclude that
\begin{equation}\label{reml}
  \Tr K^{0,0} (\sigma -\phi) (\bar \sigma-\bar\phi)  \geq  
\Tr  K^{0,0} \sigma \bar \sigma   - 2 \re\, \Tr  K^{0,0} \sigma \bar \phi\,. 
\end{equation}

The terms quadratic in $\sigma$ are thus
\begin{equation}\label{quads}
  (1- \delta )  \Tr  K^{0,0} \sigma \bar \sigma 
+ \int_{\calC\times \R^3} V(h^{-1}(x-y))|\sigma(x,y)|^2\,{dx\,dy} 
\end{equation}
with $\delta=O(h+\|\Delta\|_\infty) = O(\epsilon^{1/2} h^{1/2})$.  Pick some $\widetilde \delta
\geq 0 $ with $\delta + \widetilde \delta \leq 1/2$, and write
\begin{align}\nn
  (1- \delta ) K^{0,0} + V & = \widetilde \delta K^{0,0} + \left( 1 -
    2 \delta - 2 \widetilde \delta\right) \left( K^{0,0} + V \right) +
  \left(\delta+\widetilde \delta\right) \left(K^{0,0} + 2 V \right) \\
  & \geq \widetilde \delta K^{0,0} - C \left( \delta + \widetilde
    \delta\right)\,,
\end{align}
where we have used that $V$ is relatively form-bounded with respect to
$K^{0,0}$ to bound the last term. Hence (\ref{quads}) is bounded from
below by
\begin{equation} \widetilde \delta \left( \| \sigma\|_{H^1}^2 -
    (C-\eb+1) \|\sigma\|_2^2\right) - C \delta
  \|\sigma\|_2^2\,. \label{tdt}
\end{equation}
Recall that $\|\sigma\|_2\leq O(h^{1/2}/\epsilon)$. We shall choose
$\widetilde \delta = 0$ if the first parenthesis on the right side of
(\ref{tdt}) is less than $\tfrac 12 \|\sigma\|_{H^1}^2$ (and, in
particular, if it is negative), while $\widetilde \delta = O(1)$ in
the opposite case, i.e., when $\|\sigma\|_{H_1}^2 \geq 2
(C-\eb+1)\|\sigma\|_2^2$. In the latter case we shall have the
positive term $\widetilde\delta \|\sigma\|_{H^1}^2/2$ at our disposal,
which will be used in (\ref{upt}) below.

We are left with estimating the last term in (\ref{reml}), which is
linear in $\sigma$.  Recall from (\ref{def:sigma}) that $\sigma$ is a
sum of two terms, $\xi$ and $\sigma - \xi$, where the latter is
proportional to $\psi_>$, and $\|\xi\|_{H^1} \leq O (h^{1/2})$
independently of $\epsilon$.  Moreover, as the proof of
Theorem~\ref{lem3} shows, $\phi$ is the sum of two terms, $\eta_1$ and
$\phi-\eta_1$, with $\eta_1$ defined in (\ref{def:eta1}) (with $\psi$
replaced by $\psi_<$) and $\|\phi-\eta_1\|_{H^1} \leq O(h^{3/2})$.
Now
\begin{equation}
  \Tr  K^{0,0} \left(  \sigma -  \xi\right) \bar \eta_1 = 0
\end{equation}
as can be seen by writing out the trace in momentum space and using
that $\widehat \psi_<$ and $\widehat \psi_>$ have disjoint support.
Hence
\begin{align}\nn
  \re\, \Tr K^{0,0}  \sigma  \bar \phi & \leq C \left( \|\xi \|_{H^1}
    \|\phi\|_{H^1} + \|\sigma\|_{H^1} \|\phi -\eta_1\|_{H^1}\right)\\
  & \leq O(\epsilon h) + O(h^{3/2})\|\sigma\|_{H^1}\,.
\end{align}
In the case $\|\sigma\|_{H^1}\leq C \|\sigma\|_2$ (corresponding to
$\widetilde \delta =0$ above) we can further bound
$\|\sigma\|_{H^1}\leq O(h^{1/2}/\epsilon)$.  In the opposite case,
where $\widetilde \delta = O(1)$, we can use the positive term
$\widetilde\delta \|\sigma\|_{H^1}^2/2$ from before and bound
\begin{equation}\label{upt}
  \frac {\widetilde \delta}{2}  \|\sigma\|_{H^1}^2 -  
O(h^{3/2})\|\sigma\|_{H^1}  \geq  - O(h^{3})\,,
\end{equation}
which thus leads to an even better bound.

In combination with (\ref{lb2:s1}) these bounds show that
\begin{equation}\label{im1}
  \F^{\rm BCS}(\Gamma)   \geq h \left( \E^{\rm GP}(\psi_<) - C e \right)
\end{equation}
where
\begin{equation}
  e =  h+\epsilon^2+ \epsilon + \frac h\epsilon + \frac {h^{1/2}} {\epsilon^{3/2}}\,.
\end{equation}
The choice $\epsilon = h^{1/5}$ leads to $e \leq C h^{1/5}$.

The completes the lower bound to the BCS energy. The statement
(\ref{thm:dec}) about approximate minimizers follows immediately from
(\ref{im1}) and (\ref{defpsi}).

\bigskip

\noindent {\it Acknowledgments.} Partial financial support by the NSF
and NSERC is gratefully acknowledged.



\begin{thebibliography}{20}

\bibitem{albe} S. Albeverio, F. Gesztesy, R. H\o egh-Krohn, H. Holden,
  {\it Solvable Models in Quantum Mechanics}, $2^{\rm nd}$ ed.,
  Amer. Math. Soc. (2004)

\bibitem{BCS} J. Bardeen, L. Cooper, J. Schrieffer, {\it Theory of
    Superconductivity}, Phys. Rev. {\bf 108}, 1175--1204 (1957).

\bibitem{BDZ} I. Bloch, J. Dalibard, W. Zwerger, {\it Many-body
    physics with ultracold gases}, Rev. Mod. Phys. {\bf 80}, 885--964
  (2008).

\bibitem{zwerger} M. Drechsler, W. Zwerger, {\it Crossover from
    BCS-superconductivity to Bose-condensation}, Ann. Phys. {\bf 1},
  15--23 (1992).

\bibitem{FHNS} R.L. Frank, C. Hainzl, S. Naboko, R. Seiringer, {\it
    The critical temperature for the BCS equation at weak coupling},
  J. Geom. Anal.  {\bf 17}, 559--568 (2007).

\bibitem{FHSS} R.L. Frank, C. Hainzl, R. Seiringer, J.P. Solovej, {\it
    Microscopic derivation of Ginzburg-Landau theory}, preprint,
  arXiv:1102.4001

\bibitem{FHSS2} R.L. Frank, C. Hainzl, R. Seiringer, J.P. Solovej,
  {\it Derivation of Ginzburg-Landau theory for a one-dimensional
    system with contact interaction}, preprint, arXiv:1103.1866


\bibitem{GL} V.L. Ginzburg, L.D. Landau, {\it On the theory of
    superconductivity}, Zh. Eksp. Teor. Fiz. {\bf 20}, 1064--1082
  (1950).


\bibitem{gross} E.P. Gross, {\it Structure of a Quantized Vortex in
    Boson Systems},  Nuovo Cimento {\bf 20}, 454--466 (1961). 
{\it Hydrodynamics of a superfluid
condensate}, J. Math. Phys. {\bf 4}, 195--207 (1963).

\bibitem{HHSS} C. Hainzl, E. Hamza, R. Seiringer, J.P. Solovej, {\it
    The BCS functional for general pair interactions},
  Commun. Math. Phys. {\bf 281}, 349--367 (2008).

\bibitem{HLS} C. Hainzl, M. Lewin, \'E. S\'er\'e, {\it Existence of a
    stable polarized vacuum in the Bogoliubov-Dirac-Fock
    approximation}, Commun. Math. Phys. {\bf 257}, 515--562 (2005).

\bibitem{HS} C. Hainzl, R. Seiringer, {\it Critical temperature and
    energy gap in the BCS equation}, Phys. Rev. B {\bf 77}, 184517
  (2008).

\bibitem{HS3} C. Hainzl, R. Seiringer, {\em The BCS critical
    temperature for potentials with negative scattering length.}
  Lett. Math. Phys. {\bf 84}, 99--107 (2008).

\bibitem{Leg} A.J. Leggett, {\it Diatomic Molecules and Cooper Pairs},
  in: Modern trends in the theory of condensed matter, A. Pekalski,
  R. Przystawa, eds., Springer (1980).




\bibitem{nozieres} P. Nozi\`eres, S. Schmitt-Rink, {\it Bose
    Condensation in an Attractive Fermion Gas: From Weak to Strong
    Coupling Superconductivity}, J. Low Temp. Phys. {\bf 59}, 195--211
  (1985).

\bibitem{schl} D.S. Petrov, C. Salomon, G.V. Shlyapnikov, {\it Weakly
    Bound Dimers of Fermionic Atoms}, Phys. Rev. Lett. {\bf 93},
  090404 (2004).

\bibitem{pitaevskii} 
 L.P. Pitaevskii, {\it Vortex lines in an imperfect
    Bose gas}, Sov. Phys. JETP, {\bf 13}, 451--454 (1961).

\bibitem{randeria} M. Randeria, {\it Crossover from BCS Theory to
    Bose-Einstein Condensation}, in: Bose-Einstein Condensation,
  A. Griffin, D.W. Snoke, S. Stringari, eds., Cambridge (1995).

 \bibitem{ReSi} M. Reed, B. Simon, {\it Methods of Modern
     Mathematical Physics. IV. Analysis of Operators}, Academic
   Press (1978).

\bibitem{melo} C.A.R. S\'a de Melo, M. Randeria, J.R. Engelbrecht,
  {\it Crossover from BCS to Bose Superconductivity: Transition
    Temperature and Time-Dependent Ginzburg-Landau Theory},
  Phys. Rev. Lett. {\bf 71}, 3202--3205 (1993).

\end{thebibliography}
\end{document}